\documentclass[12pt, a4paper]{article}

\usepackage{amsmath, amsthm, amssymb}
\usepackage{geometry}
\usepackage[colorlinks=true, linkcolor=blue, citecolor=blue, urlcolor=blue]{hyperref}
\usepackage{mathtools}
\usepackage{graphicx}
\usepackage{booktabs}
\usepackage{float}
\usepackage{caption}
\usepackage{setspace}
\usepackage{titlesec}
\usepackage{parskip}
\usepackage[numbers,square,sort&compress]{natbib}

\titleformat{\section}{\large\bfseries}{}{0em}{}[\titlerule]
\titleformat{\subsection}{\normalsize\bfseries}{}{0em}{}

\newtheorem{theorem}{Theorem}
\newtheorem{lemma}[theorem]{Lemma}
\newtheorem{corollary}[theorem]{Corollary}
\newtheorem{definition}{Definition}

\title{
  \vspace{1.5cm}
  {\LARGE\bfseries Nonlinear Factor Decomposition via\\[0.3em]
  Kolmogorov-Arnold Networks:}\\[0.5em]
  {\large\itshape A Spectral Approach to Asset Return Analysis}\\[1.5cm]
}
\author{David Breazu\\[0.3em]
{\small Faculty of Mathematics and Computer Science, University of Bucharest}\\
{\small Bucharest, Romania}}
\date{March 2026}

\begin{document}

\maketitle
\thispagestyle{empty}
\newpage

\begin{abstract}
\noindent
KAN-PCA is an autoencoder that uses a KAN as encoder and a linear map as decoder. It generalizes classical PCA by replacing linear projections with learned B-spline functions on each edge. The motivation is to capture more variance than classical PCA, since the linear assumption of PCA can misrepresent the dependence structure during market stress, when correlations between assets change sharply.

We prove that if the spline activations are forced to be affine, KAN-PCA yields exactly the same result as classical PCA, which establishes PCA as a strict special case of the model. Experiments on 20 S\&P~500 stocks show that KAN-PCA captures slightly more in-sample variance than classical PCA at equal factor count ($R^2 = 65.56\%$ vs $64.71\%$ with 3 factors on the training period), while out-of-sample the two methods are close, with classical PCA marginally ahead ($57.32\%$ vs $56.73\%$). All results use a training-only benchmark that removes the data leakage present in our initial experiments.
\end{abstract}

\newpage
\tableofcontents
\newpage

\section{Introduction}

A factor model is a useful tool for understanding a portfolio of assets. Instead of analyzing every stock individually, a factor model identifies a small number of common forces that drive all assets simultaneously. These forces are called factors \citep{fama1993}. For a hedge fund, this means that instead of tracking hundreds of individual positions, one can describe the entire portfolio through a handful of underlying risk sources, which reduces a high-dimensional problem to an interpretable one.

Classical PCA is the standard method for extracting these factors from return data. However, PCA cannot capture the nonlinearities present in financial markets. During market stress, asset correlations rise sharply, so a fixed linear factor structure estimated over a full sample can misrepresent the dependence structure in such periods. PCA treats a crisis as a linear amplification of normal market dynamics rather than as a qualitatively different regime.

We therefore propose KAN-PCA: an autoencoder that uses a KAN encoder \citep{liu2024} combined with a linear decoder. The KAN encoder replaces linear projections with learned spline functions, while the linear decoder keeps the factors directly interpretable. When the spline activations are forced to be affine, KAN-PCA recovers classical PCA exactly. In the presence of genuine nonlinearity, it captures additional structure that PCA cannot represent by construction.

Our main contributions are:
\begin{enumerate}
  \item A formal proof that forcing the spline activations to be affine yields exactly the same result as classical PCA, establishing PCA as a special case of the model.
  \item Experimental evidence that KAN-PCA captures slightly more in-sample variance than classical PCA with the same number of factors ($65.56\%$ vs $64.71\%$ on the training period).
  \item An ablation study that identifies and corrects a data leakage issue in the original training procedure. Under a training-only benchmark, classical PCA and KAN-PCA are close out-of-sample, with PCA marginally ahead ($57.32\%$ vs $56.73\%$).
  \item Visualization of the learned edge functions, providing direct evidence of nonlinear factor structure in S\&P~500 returns.
\end{enumerate}

\section{Related Work}

\textbf{Nonlinear generalizations of PCA.} Classical PCA is the optimal linear method for second-order dimensionality reduction, but its linearity is a known limitation when the data lies on a nonlinear manifold. Two classical generalizations address this. Kernel PCA \citep{scholkopf1998} maps the data into a high-dimensional feature space through a fixed kernel and performs linear PCA there, so the nonlinearity is fixed in advance by the choice of kernel rather than learned from data. Autoencoders provide the second route: a neural encoder and decoder trained on a reconstruction objective learn a nonlinear low-dimensional representation. The connection between the two is classical. A linear autoencoder with squared loss recovers the PCA subspace exactly \citep{baldi1989}, and deep autoencoders were introduced as a practical nonlinear PCA by Hinton and Salakhutdinov~\citep{hinton2006}. KAN-PCA belongs to the autoencoder family, but it differs from a generic nonlinear autoencoder in that its nonlinearity is carried by learnable univariate B-spline functions on each edge, which keeps each asset's contribution individually inspectable rather than entangled in a dense weight matrix.

\textbf{Autoencoders in asset pricing.} In finance, autoencoders have been used to estimate latent factor structures from returns and characteristics. Gu et al.~\citep{gu2021} introduced conditional autoencoder asset pricing models, in which factor loadings are modeled as nonlinear functions of firm characteristics, establishing the autoencoder as a standard tool for latent factor estimation in empirical asset pricing. Our work is narrower and more classical in spirit: we extract factors directly from the return covariance structure, without conditioning on characteristics, which keeps the comparison with classical PCA exact and controlled.

\textbf{KAN-based models.} Kolmogorov-Arnold Networks \citep{liu2024} have recently been combined with autoencoders for representation learning, for example the Kolmogorov-Arnold Auto-Encoder of Yu et al.~\citep{yu2024} and the KAN autoencoders of Moradi et al.~\citep{moradi2024}, both of which report improved reconstruction over MLP autoencoders on benchmark data. Most directly related to the present work, Wang and Singh~\citep{wang2024} propose a KAN-based autoencoder for conditional asset pricing, modeling factor exposures as nonlinear functions of asset characteristics. We differ from Wang and Singh~\citep{wang2024} in two respects. First, our object of study is the covariance structure of returns rather than characteristic-conditioned exposures, which places our model in a direct, like-for-like comparison with classical PCA. Second, and central to this paper, we establish a theoretical characterization absent from prior KAN-autoencoder work: we prove that constraining the spline activations to be affine makes KAN-PCA coincide exactly with classical PCA, so that PCA is recovered as a strict special case within a larger function class. To our knowledge this exact-recovery result has not previously been stated for KAN autoencoders.

\section{Mathematical Framework}

\subsection{Classical PCA}

Let $\{x_t\}_{t=1}^T \subset \mathbb{R}^N$ be a sequence of $N$-dimensional standardized return vectors with zero mean. The empirical covariance matrix is:
\[
\Sigma := \frac{1}{T} \sum_{t=1}^{T} x_t x_t^\top \;\in\; \mathbb{R}^{N \times N}, \qquad (\Sigma_{ij} = \text{empirical covariance between asset } i \text{ and asset } j)
\]
By the spectral theorem, since $\Sigma$ is symmetric positive semidefinite, it admits the decomposition:
\[
\Sigma = U \Lambda U^\top, \qquad \Lambda = \mathrm{diag}(\lambda_1, \ldots, \lambda_N), \quad \lambda_1 \geq \cdots \geq \lambda_N \geq 0
\]
where $U = [u_1 \mid \cdots \mid u_N]$ satisfies $U^\top U = I$ (the columns are orthonormal: each eigenvector has unit length and is perpendicular to all others). The $k$-factor PCA projection uses the top-$k$ eigenvectors $U_k := [u_1 \mid \cdots \mid u_k]$:
\[
z_t = U_k^\top x_t \;\in\; \mathbb{R}^k, \qquad \hat{x}_t = U_k z_t \;\in\; \mathbb{R}^N
\]
The reconstruction loss equals $\sum_{i=k+1}^{N} \lambda_i$, the minimum achievable among all rank-$k$ linear approximations. On the full sample with $N=20$ stocks and $k=3$, the leading eigenvalues account for $R^2 = 62.99\%$ of the total variance, with:
\[
\lambda_1 \text{ explains } 48.00\%, \quad \lambda_2 \text{ explains } 8.96\%, \quad \lambda_3 \text{ explains } 6.03\%
\]
(The experiments in Section~5 instead fit PCA on the training period only, for a fair comparison against the learned models, which gives a slightly different in-sample figure.)

\subsection{B-Spline Functions}

\begin{definition}[B-Spline Space]
Let $G = \{t_0 < t_1 < \cdots < t_m\} \subset \mathbb{R}$ be a knot vector and $k \geq 1$. The B-spline space of degree $k$ over $G$ is:
\[
\mathcal{S}_k(G) := \left\{ \varphi \in C^{k-1}(\mathbb{R}) \;\middle|\; \varphi\big|_{[t_i, t_{i+1}]} \in \mathcal{P}_k,\; \forall i \right\}
\]
where $\mathcal{P}_k$ denotes the space of polynomials of degree at most $k$ \citep{deboor2001}.
\end{definition}

We use $k$ to denote the polynomial degree, so a degree-$k$ spline has $C^{k-1}$ continuity at simple interior knots (equivalently, order $k+1$ in de~Boor's convention). B-splines are piecewise polynomials that are smooth at their knots. A cubic B-spline ($k=3$) is therefore $C^2$-continuous, which is sufficient smoothness for our purposes. The basis functions $\{B_{i,k}\}$ are computed via the Cox-de~Boor recursion and satisfy the partition of unity property $\sum_i B_{i,k}(x) = 1$ \citep{deboor2001}.

\subsection{Kolmogorov-Arnold Networks}

The Kolmogorov-Arnold representation theorem \citep{kolmogorov1957} states that any continuous function $f: [0,1]^n \to \mathbb{R}$ can be written as:
\[
f(x_1, \ldots, x_n) = \sum_{q=1}^{2n+1} \Phi_q\!\left( \sum_{p=1}^{n} \varphi_{q,p}(x_p) \right)
\]
where $\Phi_q, \varphi_{q,p}: \mathbb{R} \to \mathbb{R}$ are continuous univariate functions.

\begin{definition}[KAN Layer]
A KAN layer $\Phi : \mathbb{R}^{n} \to \mathbb{R}^{m}$ is defined componentwise by:
\[
\left(\Phi(x)\right)_j := \sum_{i=1}^{n} \varphi_{ji}(x_i), \qquad j = 1, \ldots, m
\]
where each $\varphi_{ji} \in \mathcal{S}_k(G)$ is a learnable univariate B-spline function \citep{liu2024}. The full activation combines a SiLU \citep{elfwing2018} basis term with the spline:
\[
\varphi_{ji}(x) = w_b \cdot \mathrm{SiLU}(x) + w_s \cdot \sum_{l} c_l B_{l,k}(x)
\]
with learnable parameters $w_b, w_s \in \mathbb{R}$ and spline coefficients $c_l \in \mathbb{R}$.
\end{definition}

The key distinction from MLPs is that the activations reside on edges. Each edge carries its own learned function rather than a scalar weight, and the nodes merely sum their inputs.

\subsection{KAN-PCA}

\begin{definition}[KAN-PCA]
KAN-PCA is a KAN autoencoder consisting of:
\begin{itemize}
  \item A KAN encoder $\mathcal{E} : \mathbb{R}^N \to \mathbb{R}^k$, a composition of KAN layers;
  \item A linear decoder $\mathcal{D} : \mathbb{R}^k \to \mathbb{R}^N$, defined by $\mathcal{D}(z) = Wz$ where $W \in \mathbb{R}^{N \times k}$ is the decoding matrix: row $i$ of $W$ specifies how asset $i$ is reconstructed from the $k$ factors;
\end{itemize}
trained to minimize the reconstruction loss:
\[
\mathcal{L}(\mathcal{E}, \mathcal{D}) := \frac{1}{T} \sum_{t=1}^{T} \left\| x_t - \mathcal{D}\!\left(\mathcal{E}(x_t)\right) \right\|^2
\]
The $k$-dimensional representations $z_t = \mathcal{E}(x_t) \in \mathbb{R}^k$ are the nonlinear factors: the compressed description of the market on day $t$.
\end{definition}

Unlike classical PCA, each asset's contribution to each factor is mediated by a learned spline function $\varphi_{ji}$ rather than a fixed scalar weight. This lets the model capture asymmetric and regime-dependent relationships between assets and factors.

\section{Theoretical Result}

\begin{lemma}[Affine Splines are Linear Maps]
\label{lem:linear}
A spline function $\varphi \in \mathcal{S}_k(G)$ whose second derivative vanishes identically is affine, $\varphi(x) = wx + b$ for some $w, b \in \mathbb{R}$. For zero-mean standardized inputs the intercept $b$ does not affect the recovered factor subspace, so it may be set to $b = 0$ without loss of generality, giving $\varphi(x) = wx$.
\end{lemma}

\begin{proof}
A degree-1 polynomial is automatically an element of any spline space $\mathcal{S}_k(G)$ with $k \geq 1$. If the second derivative is zero everywhere, the function does not curve at any point, so it is affine: $\varphi(x) = wx + b$. When the inputs $x_t$ are centered to zero mean, a constant shift $b$ applied on each edge contributes only a fixed offset to the layer output and is absorbed without changing the covariance structure or the recovered subspace; we therefore take $b = 0$.
\end{proof}

\begin{lemma}[Linear KAN Layer is Matrix Multiplication]
\label{lem:kan-linear}
If all spline functions $\varphi_{ji} \in \mathcal{S}_k(G)$ satisfy $\varphi_{ji}'' \equiv 0$ for all $i, j$, then the KAN layer $\Phi : \mathbb{R}^n \to \mathbb{R}^m$ reduces to:
\[
\Phi(x) = Wx, \qquad W_{ji} := w_{ji} \in \mathbb{R}
\]
\end{lemma}

\begin{proof}
By Lemma~\ref{lem:linear}, $\varphi_{ji}(x_i) = w_{ji} x_i$ for each edge $(j,i)$. Therefore:
\[
\left(\Phi(x)\right)_j = \sum_{i=1}^{n} \varphi_{ji}(x_i) = \sum_{i=1}^{n} w_{ji} x_i = (Wx)_j
\]
where $W = (w_{ji}) \in \mathbb{R}^{m \times n}$. The full layer computes $x \mapsto Wx$.

\medskip
\noindent\textit{Concrete example (3 stocks, 3 nodes):}
\[
\begin{pmatrix} \mathrm{node}_1 \\ \mathrm{node}_2 \\ \mathrm{node}_3 \end{pmatrix}
=
\begin{pmatrix} w_{11} & w_{12} & w_{13} \\ w_{21} & w_{22} & w_{23} \\ w_{31} & w_{32} & w_{33} \end{pmatrix}
\begin{pmatrix} x_{\mathrm{AAPL}} \\ x_{\mathrm{MSFT}} \\ x_{\mathrm{GOOG}} \end{pmatrix}
= Wx
\]
\end{proof}

\begin{lemma}[Optimal Linear Autoencoder $=$ PCA~\citep{baldi1989}]
\label{lem:baldi}
Let $\Sigma = U\Lambda U^\top$ be the spectral decomposition of the empirical covariance matrix. The linear autoencoder:
\[
\min_{A \in \mathbb{R}^{k \times N},\; W \in \mathbb{R}^{N \times k}} \frac{1}{T}\sum_{t=1}^T \|x_t - WAx_t\|^2
\]
achieves its global minimum at $A = U_k^\top$, $W = U_k$, where $U_k$ are the top-$k$ eigenvectors of $\Sigma$. The minimum value is $\sum_{i=k+1}^{N} \lambda_i$.
\end{lemma}

\begin{proof}
Expanding the loss using $\Sigma = \frac{1}{T}\sum_t x_t x_t^\top$:
\[
\mathcal{L}(A, W) = \mathrm{tr}\!\left(\Sigma - \Sigma A^\top W^\top - WA\Sigma + WA\Sigma A^\top W^\top\right)
\]
Setting $\frac{\partial \mathcal{L}}{\partial W} = 0$ gives $W = \Sigma A^\top (A\Sigma A^\top)^{-1}$. Substituting back:
\[
\mathcal{L}(A) = \mathrm{tr}(\Sigma) - \mathrm{tr}\!\left(A\Sigma A^\top(AA^\top)^{-1}\right)
\]
By the Ky Fan maximum principle \citep{fan1949}, the maximum of $\mathrm{tr}(A\Sigma A^\top(AA^\top)^{-1})$ is attained when the rows of $A$ span the top-$k$ eigenspace of $\Sigma$, giving $A = U_k^\top$ and $W = U_k$. The minimum loss is $\mathcal{L}^* = \sum_{i=k+1}^N \lambda_i$.
\end{proof}

\begin{theorem}[KAN-PCA Reduces to PCA in the Linear Limit]
\label{thm:main}
Let KAN-PCA be a KAN autoencoder $\mathcal{E}$ with linear decoder $\mathcal{D}$, trained to minimize the reconstruction loss. If all spline activations $\varphi_{ji}$ are constrained to satisfy $\varphi_{ji}'' \equiv 0$ everywhere, then the optimal encoder and decoder satisfy:
\[
\mathcal{E}^*(x) = U_k^\top x \qquad \mathcal{D}^*(z) = U_k z
\]
where $U_k^\top x$ is a standard matrix multiplication, the linear limit of the KAN encoder by Lemmas~\ref{lem:linear} and~\ref{lem:kan-linear}. Consequently:
\[
\mathcal{L}^*_{\mathrm{KAN}} \;\leq\; \mathcal{L}^*_{\mathrm{PCA}} = \sum_{i=k+1}^{N} \lambda_i
\]
with equality if and only if the data admits a purely linear factor structure.
\end{theorem}

\begin{proof}
By Lemmas~\ref{lem:linear} and~\ref{lem:kan-linear}, the KAN encoder reduces to $\mathcal{E}_{\mathrm{lin}}(x) = Ax$. The loss becomes the linear autoencoder objective of Lemma~\ref{lem:baldi}, giving $A^* = U_k^\top$ and $W^* = U_k$.
\end{proof}

\begin{corollary}
Classical PCA is a special case of KAN-PCA:
\[
\mathrm{KAN\text{-}PCA}\big|_{\varphi_{ji}'' \equiv 0} \;=\; \mathrm{PCA}
\]
Writing $\mathcal{L}_{\mathrm{aff}}$ for the affine functions, the inclusion $\mathcal{L}_{\mathrm{aff}} \subsetneq \mathcal{S}_k(G)$ for $k \geq 2$ is strict, so KAN-PCA optimizes over a strictly larger function class than PCA.
\end{corollary}

\section{Experiments}

\subsection{Data and Setup}

We use daily log-returns for 20 S\&P~500 large-cap constituents from January 2015 to January 2024, comprising 2,263 trading days. The tickers are AAPL, MSFT, GOOGL, JPM, JNJ, XOM, BAC, PG, KO, WMT, CVX, GS, HD, MCD, CAT, MMM, IBM, GE, BA, and DIS. Returns are standardized to zero mean and unit variance per stock.

The dataset is split 70/10/20: a training set (1,584 days, 2015--2020), a validation set (226 days, 2021), and a test set (453 days, 2022--2024). All hyperparameter choices are made on the training and validation data only.

\subsection{Models}

We benchmark KAN-PCA against classical PCA, all using $k=3$ factors. Classical PCA is the standard spectral decomposition of $\Sigma$, fitted on the training period only to give a fair, leakage-free benchmark. The KAN-PCA encoder uses the architecture $[20, 10, 3, 20]$: 20 stock inputs, an intermediate KAN layer of 10 neurons, compression to 3 factors, and a linear decoder back to 20 outputs.

We arrived at the final configuration through several iterations. Initial variants (a simpler $[20,3,20]$ encoder and a deeper $[20,10,3,20]$ encoder with grid $G=10$) were trained without a proper validation split, which, as we explain in the ablation study, led to data leakage and optimistic out-of-sample numbers. The final model, KAN-PCA Improved, corrects this with a proper 70/10/20 split, early stopping on the validation loss, progressive grid extension ($3 \to 5 \to 10$), and mild regularization (spline penalty $\lambda$ decreasing across stages from $10^{-3}$ to $10^{-4}$, entropy penalty $0.1$). All numbers reported below are from this corrected procedure.

\subsection{Results}

\begin{table}[H]
\centering
\caption{Classical PCA versus KAN-PCA on 20 S\&P~500 stocks ($k=3$ factors, 70/10/20 split). Both models are fitted on the training period only.}
\begin{tabular}{lcc}
\toprule
Model & $R^2$ Train & $R^2$ Out-of-Sample \\
\midrule
Classical PCA (train-fit) & 64.71\% & \textbf{57.32\%} \\
KAN-PCA Improved & \textbf{65.56\%} & 56.73\% \\
\bottomrule
\end{tabular}
\end{table}

\noindent\textit{$R^2$ Train $=$ in-sample explained variance on the training period. $R^2$ Out-of-Sample $=$ explained variance on the held-out test period.}

\medskip
On the training period, KAN-PCA captures slightly more variance than classical PCA ($65.56\%$ vs $64.71\%$), which indicates that some nonlinear structure exists in the returns that a purely linear projection cannot represent. Out-of-sample the two methods are close, with classical PCA marginally ahead by $0.59$ percentage points ($57.32\%$ vs $56.73\%$). In other words, the additional flexibility of the spline activations helps in-sample but does not translate into an out-of-sample advantage on this dataset and universe size. This is consistent with the theoretical result of Section~3: when the underlying factor structure is close to linear, the larger function class of KAN-PCA has little genuine nonlinearity to exploit, and it reduces to behaviour very near that of PCA.

\subsection{Ablation Study}

The first versions of the model were trained and evaluated without a proper separation between training and test data. This created a data leakage issue: the pykan training routine uses the evaluation set for internal grid updates, so the model indirectly observed test-set statistics while learning. As a result, the early out-of-sample numbers were optimistic and not reproducible under a clean protocol.

The corrected model uses a proper 70/10/20 train/validation/test split, early stopping on the validation loss, progressive grid extension ($3 \to 5 \to 10$), and mild regularization. Under this protocol the model trains stably and the comparison against PCA is fair: KAN-PCA is $0.59$pp below classical PCA out-of-sample ($56.73\%$ vs $57.32\%$) while being slightly ahead in-sample. The lesson of the ablation is methodological. Most of the apparent out-of-sample strength in the first experiments was an artifact of leakage rather than a real property of the model, and a careful protocol is essential when evaluating flexible nonlinear models on financial time series.

\subsection{Factor Analysis}

\begin{figure}[H]
\centering
\includegraphics[width=\textwidth]{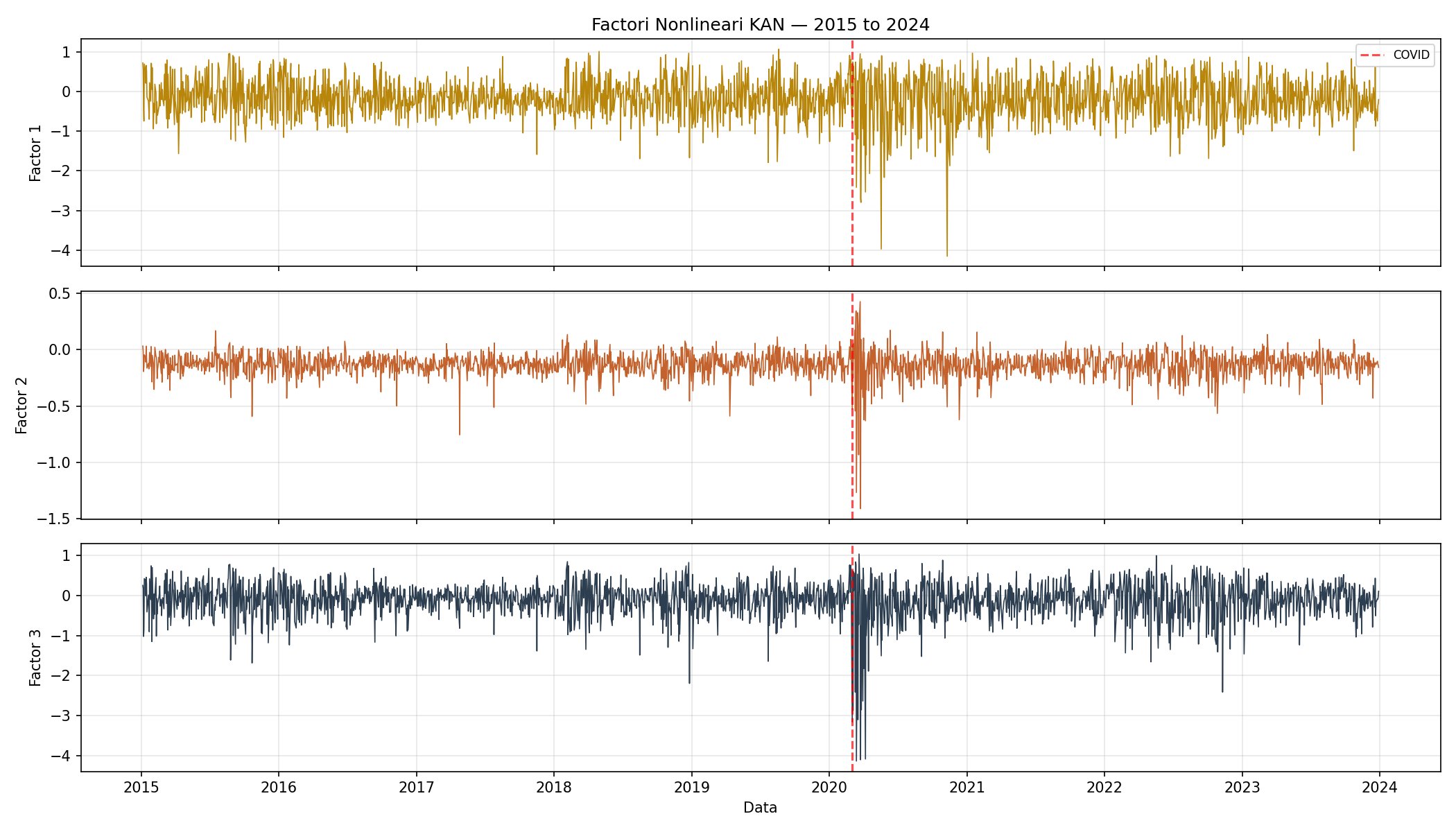}
\caption{KAN nonlinear factors (2015--2024). The red dashed line marks the COVID-19 crash of March 2020. All three factors show sharp movements at this date, in contrast to the linear response of classical PCA.}
\end{figure}

All three KAN factors react sharply at the COVID-19 crash. On the standardized scale, each factor drops well below its normal range during March 2020: Factor~1 and Factor~3 fall toward $-4$, and Factor~2 toward $-1.5$, in each case several times their typical daily amplitude. The qualitative point is that the model registers the crash as a large, abrupt deviation in all three coordinates simultaneously, whereas classical PCA spreads the same shock linearly across its components. The magnitudes themselves should be read as relative excursions rather than calibrated quantities, since the factors are only identified up to scale.

\subsection{Edge Function Analysis}

\begin{figure}[H]
\centering
\includegraphics[width=\textwidth]{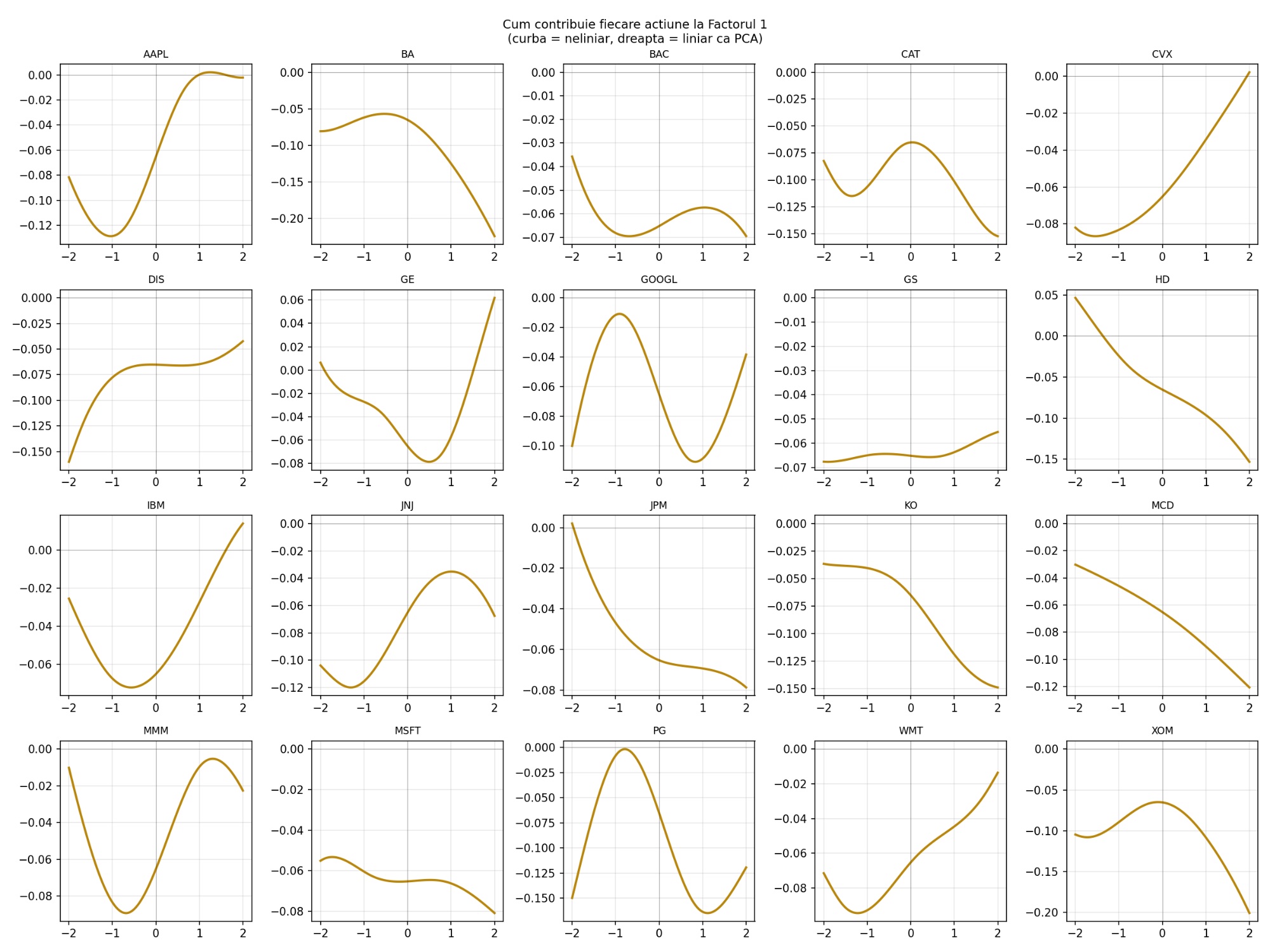}
\caption{Learned edge functions $\varphi(x)$ for all 20 stocks (contribution to Factor 1). If the relationships were linear, as PCA assumes, all curves would be straight lines. They are not, which gives direct visual evidence of nonlinear factor structure.}
\end{figure}

The learned edge functions $\varphi_{ji}(x)$ are visibly nonlinear across the 20 stocks; a strictly linear PCA loading would appear as a straight line in every panel. Several patterns stand out. DIS shows a sharp localized peak near zero, so small returns contribute very differently from a linear extrapolation. HD is nearly flat for negative returns and then rises steeply for positive ones, an asymmetric response. PG is bimodal, with two distinct local minima, consistent with more than one operating regime. Other names such as KO, MCD, and XOM are closer to monotone but still curved, indicating mild rather than strong nonlinearity. Taken together, the panels show that the nonlinearity captured by the model is real but uneven across assets.

\section{Conclusion}

In this paper we studied KAN-PCA, a factor model that replaces the linear projections of classical PCA with learned spline functions on each edge while keeping a linear, interpretable decoder. Our main result is theoretical: when the spline activations are constrained to be affine, KAN-PCA coincides exactly with classical PCA, so PCA is a strict special case within a larger function class. Empirically, on 20 S\&P~500 stocks the model captures slightly more variance than PCA in-sample ($65.56\%$ vs $64.71\%$), while out-of-sample the two are close, with PCA marginally ahead ($57.32\%$ vs $56.73\%$). The learned edge functions are visibly nonlinear, which gives direct evidence that some nonlinear structure is present even though it does not translate into an out-of-sample gain at this universe size.

We also reported a methodological lesson. Our first experiments suffered from data leakage through the grid-update routine, which produced optimistic out-of-sample numbers; a proper validation split removes this effect and yields the fair comparison above.

A natural next step is to apply the same construction to a larger universe of assets, on the order of 200 stocks, with a deeper architecture and a longer training procedure. Our preliminary attempts at that scale did not train stably, and we leave a careful treatment, including the sparsification and grid-extension techniques of KAN~2.0~\citep{liu2024v2}, to future work. Whether the larger function class of KAN-PCA yields a genuine out-of-sample advantage on returns data, rather than only an in-sample one, remains an open empirical question.

\bibliographystyle{unsrtnat}

\end{document}